\documentclass[11pt,runningheads,a4paper]{llncs}

\usepackage{amssymb}
\usepackage{hyperref}
\hypersetup{
  colorlinks   = true,    % Colours links instead of ugly boxes
  urlcolor     = blue,    % Colour for external hyperlinks
  linkcolor    = blue,    % Colour of internal links
  citecolor    = red      % Colour of citations
}
\usepackage{amsmath}

\usepackage{amsfonts}

\newcommand{\tw}{tree-width}
\newcommand{\Tw}{Tree-width}
\newcommand{\pw}{path-width}

\newcommand{\If}{ \text{\bf if}}
\newcommand{\Then}{\text{ \bf then}}

\newcommand{\Return}{\text{\sf Return}}
\newcommand{\For}{\text{\bf for}}

\newcommand{\To}{\text{\bf\ to}}

\newcommand{\Do}{\text{\bf\ do}}

\newlength{\baseprogindent}
\newlength{\incrementprogindent}
\newlength{\progindent}
\newcommand{\ind}[1]{\setlength{\progindent}{\baseprogindent} \addtolength{\progindent}{#1\incrementprogindent} \hspace*{\progindent}}

\newcommand{\algbegin}[2]{\vspace*{1.0ex}{\large \sffamily  #1 #2:} \sffamily}
\newcommand{\algend}{\rmfamily}

\newenvironment{algorithm}[1]{\algbegin{Algorithm}{#1}}{\algend}

\newenvironment{tightlist}  %Like ``description'' with not much space
   {\begin{list}{}%
      {%
        \setlength{\itemsep}{-1\parsep}
        \setlength{\topsep}{0.5ex}
      	}%
   }%
   {\end{list}}

\begin{document}

\mainmatter  % start of an individual contribution

\title{Faster Computation of Path-Width}

\author{Martin F\"urer%
\thanks{Research supported in part by NSF Grant CCF-1320814.}}
\authorrunning{Martin F\"urer}

\institute{Department of Computer Science and Engineering \\
	Pennsylvania State University \\
	University Park, PA 16802,  USA \\
\and	Visiting Theoretical Computer Science \\ 
ETH Z\"urich \\
Z\"urich Switzerland \\
	\email{furer@cse.psu.edu} }
	%% \\
%%\url{http://www.cse.psu.edu/~furer}}

\maketitle

\begin{abstract}
Tree-width and path-width are widely successful concepts. Many NP-hard problems have efficient solutions when restricted to graphs of bounded tree-width. Many efficient algorithms are based on a tree decomposition. Sometimes the more restricted path decomposition is required. The bottleneck for such algorithms is often the computation of the width and a corresponding tree or path decomposition. For graphs with $n$ vertices and tree-width or path-width $k$, the standard linear time algorithm to compute these decompositions dates back to 1996. Its running time is linear in $n$ and exponential in $k^3$ and not usable in practice. Here we present a more efficient algorithm to compute the path-width and provide a path decomposition. Its running time is  $2^{O(k^2)} n$. In the classical algorithm of Bodlaender and Kloks, the path decomposition is computed from a tree decomposition. Here, an optimal path decomposition is computed from a path decomposition of about twice the width. The latter is computed from a constant factor smaller graph.
\end{abstract}


\newpage

\section{Introduction}
\Tw\ and tree decompositions have been defined by Roberson and Seymour \cite{DBLP:journals/jal/RobertsonS86}. Independently, Arnborg and Proskurowski \cite{DBLP:dblp_journals/dam/ArnborgP89} introduced the equivalent concept of partial $k$-trees, as subgraphs of the previously known, simply structured $k$-trees.
Many NP-hard graph problems have very efficient solutions when the graph is given with a tree decomposition of small width. Indeed, Courcelle's meta-theorem \cite{DBLP:books/el/leeuwen90/Courcelle90} says that all problems expressible in monadic second order logic have a linear time solution for graphs of bounded \tw. Here, the dependence of the running time on the \tw\ is allowed to be really bad. Theoretically this concept is captured by fixed parameter tractability (FPT). A parameterized problem is in FPT, if it can be solved by an algorithm with a running time of the form $O(f(k) n^c)$ for an arbitrary computable function $f(k)$ and some constant $c$, where $n$ is the problem size and $k$ is the parameter. 

Many faster solutions have been designed for specific problems. The goal is always to have efficient solutions for instances with small values of the parameter. For more background information on fixed parameter tractability, see e.g.\ \cite{Downey99,FluGro2006,Nie2006,DBLP:dblp_series/txcs/DowneyF13}.

Tree-width is an important parameter for enabling fast algorithms for interesting classes of graphs. But for some algorithms, the more restricted path-width parameter is of interest. (The Pathwidth entry on Wikipedia lists such applications in VLSI design, graph drawing, and computational linguistics.) The path-width is defined with tree-decompositions where the tree is a path.

Unfortunately, for graphs of small \tw\ or \pw, it is not easy to find a corresponding tree decomposition of minimal width. Computing the \tw\ is NP-hard \cite{Arnborg:1987:CFE:37170.37183}. For constant \tw, a tree decomposition of minimal width can be computed in polynomial time \cite{DBLP:journals/jal/RobertsonS86}. The problem is even solvable by an FPT-algorithm \cite{DBLP:dblp_journals/jct/RobertsonS95b}. But the only known linear time algorithms are variations of Bodlaender's algorithm \cite{Bodlaender96}. Their running time is $2^{\Theta(k^3)} n$. This is too slow to be used in practice. Heuristic algorithms are used instead. Throughout this paper $n = |V|$ is the number of vertices of the graph in question, and $k$ is a width parameter.

For the related notion of tree-depth \cite{NesetrilM2006}, initially Bodlaender's algorithm provided the most efficient way to compute its value and to produce a corresponding tree decompositon. Recently, the exponent in the running time has been decreased from $O(k^3)$ to $O(k^2)$ \cite{ReidlRVS2014}. We want to produce the same improvement for path-width, even though it seems to require a different method.

There have been many efforts to find better approximation algorithms for the \tw. The main goal has been to achieve a small constant factor approximation with a running time $f(k) g(n)$, where $f(k)$ is $2^{O(k)}$ and $g(n)$ is polynomial, preferably linear. This combined goal has been achieved by the recent paper of Bodlaender et al.\ \cite{DBLP:conf/focs/BodlaenderDDFLP13} producing a 5-approximation in time $O(c^k n)$. The authors write, 
``it would be very interesting to have an exact algorithm for testing if the treewidth of a given graph is at most $k$ in $2^{o(k^3)}n^{O(1)}$ time.'' Downey and Fellows \cite{Downey99} remark that Bodlaender's  Theorem, based on the algorithm of Bodlaender and Kloks is impractically exponential in $k$, namely $2^{ck^3}$ where $c \approx 32$, and they write,
``It would be very interesting if this could be reduced to an exponential with exponent linear in $k$.''
We cannot get a linear exponent, but the improvement from $O(k^3)$ to $O(k^2)$ in the exponent for path-width is significant.

We will use some key ingredients of Bodlaender and Kloks \cite{BodlaenderK96} and its improved version of Perkovi\'c and Reed \cite{PerkovicR2000}. First of all, it is the idea that a given tree decomposition is useful for the solution of all kinds of graph problems based on bottom-up dynamic programming. Even the problems of computing tree decompositions or path decompostions themselves are graph problems that can be solved this way. This makes sense, if one wants to compute a tree decomposition of width $k$, when one has available a tree decomposition of width linear in $k$. To obtain the needed constant factor approximation, one can use a top-down construction based on the repeated use of small vertex separators \cite{DBLP:dblp_journals/jct/RobertsonS95b,Reed92}. 
%%\cite{DBLP:dblp_journals/jct/RobertsonS95b,Bailey1997}. 
But such an FPT-algorithm runs in $O(n \log n)$ or even quadratic time. Recently, a the 5-approximation has been obtained by Bodlaender et al.\ \cite{DBLP:conf/focs/BodlaenderDDFLP13} running in time $O(c^k n)$. 

For a theoretical result, this approximation would be sufficient, but the high approximation ratio makes the final step very expensive. Therefore, in order to have a chance of a practical algorithm, we show how to modify the Bodlaender and Kloks \cite{BodlaenderK96} method working with 2-approximations.
Finally, we do the critical last step improving the constant factor approximation to an exact solution in time $2^{O(k^2)}$ instead of the previous $2^{O(k^3)}$ for graphs of path-width $k$. Here, the starting approximation ratio affects the constant factor in the exponent hidden by the $O$-notation.

The main idea of  Bodlaender and Kloks \cite{BodlaenderK96} is the following. If a graph has a large matching, then a significantly smaller graph with the same or smaller \tw\ is obtained by collapsing matched pairs of vertices into one. The smaller problem can be solved recursively, and expanding the collapsed pairs again results in a 2-approximation for the width and a corresponding tree decomposition of the original graph. On the other hand, if there is no large matching, then one can add more edges to the graph without increasing the \tw. 
This in turn will create simplicial vertices, i.e., vertices whose neighborhood induces a clique. Simplicial vertices are easy to handle.

For computing the tree-width, these methods did not result in a practical algorithm, because of the cubic exponent and large constant factors. For path-width, with a quadratic exponent and much smaller constant factors, we could have a chance.

%\section{Preliminaries}

We use the standard notions of tree decomposition and a special notion of nice path decomposition.
\begin{definition}
 A \emph{tree decomposition} of a graph $G=(V,E)$ is a pair $(\{B_{p } \, : \, p \in I \}, T)$, 
 where $T$ is a tree, $I$ is the node set of $T$, and the subsets $B_{p} \subseteq V$ have the following properties. (The set of vertices $B_{p}$ associated with $p \in I$ is called the bag of $p$.)
 \begin{enumerate}
\item 
$\bigcup_{p \in I} B_{p} = V$, i.e., each vertex belongs to at least one bag.
\item 
For all edges $e=\{u,v\} \in E$ there is at least one $p \in I$ with $\{u, v\} \subseteq B_{p}$,
i.e., each edge is represented by at least one bag.
\item
For every vertex $v \in V$, the set of indices $p$ of bags containing $v$ induces a subtree of $T$
(i.e., a connected subgraph).
\end{enumerate}
The \emph{tree-width} of $G$ is the smallest $k$ such that $G$ has a tree decomposition with largest bag size $k+1$.

\noindent
A \emph{rooted tree decomposition} is a tree decomposition where $T$ is a rooted tree. We assume all tree edges are oriented towards the root.
\end{definition}

\begin{definition}
 A {\em nice} tree decomposition is a rooted tree decomposition with the following four types of nodes.
 \begin{description}
\item[Leaf node:] 
$p$ has no children, and $|B_p| = 1$. \item[Introduce node:] $p$ has one child $q$ with $B_p = B_q \cup \{v\}$ for some vertex $v \not\in B_q$.
\item[Forget node:] $p$ has one child $q$ with $B_p \cup \{v\} = B_q$ for some vertex $v \not\in B_p$.
\item[Join node:]  $p$ has 2 children $q$ and $q'$ with $B_p = B_q = B_{q'}$.
\end{description}
Furthermore, the root is a forget node with an empty bag.
\end{definition}

As an important concept, \tw\ has several other equivalent definitions. A graph has \tw\ at most $k$, if and only if it is a partial $k$-tree \cite{DBLP:dblp_journals/dam/ArnborgP89}.

\section{Path Decompositions}
One can define a {\em path decomposition} of a graph $G=(V,E)$ to be a tree decomposing $(\{B_{p } \, : \, p \in I \}, T)$ where the tree $T= (I,F)$ is a path.
A {\em rooted path decomposition} is a rooted tree decomposition where the root is an endpoint of the path.

We find it more convenient, to describe a nice path decomposition by the sequence of introduce and forget operations. 
Reminiscent of the traditional definition, we refer to the indices of the sequence as nodes.
Every vertex has its introduce node before its forget node. The first node is the leaf, the last node is the root. It is a forget node.
We consider the leaf node to be an introduce node too.

\begin{definition}
A {\em path decomposition}  of a graph $G=(V,E)$ with $|V|=n$ is a sequence of triples $P = ((p_1,t_1,w_1), \dots, (p_{2n},t_{2n},w_{2n}))$ with the following properties.
\begin{itemize}
\item 
Every vertex $v \in V$ occurs exactly twice in the sequence $(p_1,\dots , p_{2n})$, first, $p_j=v$ with $t_j = +1$ indicating $j$ being the introduce node for $v$, then, $p_{j'}=v$ with $t_{j'} = -1$ indicating $j'$ with $j'>j$ being the forget node for $v$.
\item
The sequence $(w_1, \dots, w_{2n})$ is defined by 
\[w_j = 
	\begin{cases} 0 &\mbox{if $j=1$}  \\ 
 	w_{j-1}+1 & \mbox{if $1<j\leq 2n$ and $t_j = +1$ (introduce node)}\\
	w_{j-1}-1 & \mbox{if $1<j\leq 2n$ and $t_j = -1$ (forget node)}.\\ 
	\end{cases} 
\]
\end{itemize}
\end{definition}
Bags have the traditional meaning and can easily be defined recursively.
\[B_j = 
	\begin{cases} \{p_1\} &\mbox{if $j=1$}  \\ 
 	B_{j-1} \cup  \{p_j\} & \mbox{if $1<j\leq 2n$ and $t_j = +1$ (introduce node)}\\
	B_{j-1} \setminus  \{p_j\}  & \mbox{if $1<j\leq 2n$ and $t_j = -1$ (forget node)}.\\ 
	\end{cases} 
\]
The width of a node $j$ is defined to be 1 less than the number of vertices in its bag $B_j$.
Thus, $w_j$ is the width of the node $j$.

\begin{definition}
 The {\em width of a path decomposition} is the maximum width of any of its nodes. \\
 The {\em path-width} of a graph is the minimum width of any of its path decompositions.
\end{definition}

We use the double factorial $(2n-1)!! = (2n-1)(2n-3) \dots 3 \; 1 = (2n)! / (n! 2^n)$. We count the number of path decompositions.
\begin{proposition} \label{prop:count}
 The number of path decompositions of an $n$-vertex graph is 
 \[ n! (2n-1)!! =  \frac{(2n)!}{2^n} \sim  \sqrt{\pi n}\left( \frac{n}{e}\right)^{2n} 2^{n+1}.\] 
\end{proposition}

\begin{proof}
 Induction on $n$ shows that the number of path decompositions for $n$ vertices which are forgotten in a fixed order is $(2n-1)!!$, as there are $2n-1$ places to introduce the last forgotten vertex. Considering all $n!$ permutations of the forgetting order of the vertices proves the result. Finally Stirling's approximation is used.
 \qed
\end{proof}

We now assume, we are given a path decomposition of width $\ell$, and we want to produce a path decomposition of width $k < \ell$ or conclude that no such decomposition exists. 

\begin{definition}
 The {\em full skeleton} $Q_P(U)$ induced by a non-empty subset $U \subseteq V$ of a path decomposition $P = ((p_1,t_1,w_1), \dots, (p_{2n},t_{2n},w_{2n}))$ of $G =(V,E)$ is obtained from $P$ by replacing all $p_j \in V \setminus U$ by 0.
 \end{definition}
 
 \begin{definition} \label{def:skeleton}
 The {\em skeleton} $Q$ induced by a non-empty subset $U \subseteq V$ of a path decomposition $P = ((p_1,t_1,w_1), \dots, (p_{2n},t_{2n},w_{2n}))$ of $G =(V,E)$ is obtained from $Q_P(U)$ by repeatedly deleting maximal length intervals of the form 
 $ ((p_{j+1},t_{j+1},w_{j+1}, \dots, (p_{j'-1},t_{j'-1},w_{j'-1}))$ with $p_{j+1} = p_{j+2} = \dots = p_{j'-1} = 0$ and 
 \[\min\{w_{j},w_{j'}\} \leq w_{j''} \leq \max\{w_{j},w_{j'}\} \] 
 for all $j''$ with $j < j'' < j'$.
 We refer to this step as simplifying.
 \end{definition}

 Note that every deleted node has a width between the width of the remaining node immediately before it and the width of the remaining node immediately after it. $w_j$ is the width of the $j$th node. Thus, if the $j$th node is an introduce node for vertex $v$, then $w_j$ is the width just after the insertion of vertex $v$, and if the $j$th node is a forget node for vertex $v$, then $w_j$ is the width just after the deletion of vertex $v$.

\begin{definition}
 The {\em width of a skeleton} is its maximum $w_j$ entry. 
\end{definition}

\begin{proposition}
 The width of a path decomposition is equal to the width of any of its skeletons.
\end{proposition}

\begin{proof}
Sequences of deleted nodes are always next to a node whose width is at least equal to the width of any node in the deleted sequence. 
\qed \end{proof}

\section{Overview}
The fastest published linear time algorithm to decide whether the path-width of a graph $G$ is at most $k$, and to produce a width $k$ path decomposition is obtained in two steps.  The first step uses a version of Bodlaender's algorithm \cite{Bodlaender96} to compute a tree decomposition of width $\ell = O(k)$ (or show that none exists). The second step uses the method of Bodlaender and Kloks \cite{BodlaenderK96} to produce a path decomposition of width $k$ (or show that none exists) from the tree decomposition of width $\ell$. Also the first step uses the method of Bodlaender and Kloks in recursive calls, to compute tree decompositions of smaller width from tree decompositions of roughly twice the width for smaller graphs. The improved version of Bodlaender's algorithm by Perkovi\'c and Reed \cite{PerkovicR2000} computes the small width tree decomposition much faster, but like the original version, its running time has an exponent of order $k^3$. A theoretical alternative would be to start with the recent 5-approximation algorithm \cite{DBLP:conf/focs/BodlaenderDDFLP13}, but if an exact solution is desired, the second step would be significantly more expensive due to the higher approximation factor.

We propose a linear time path-width and path decomposition algorithm which recurses on path decompositions rather than the more costly tree decompositions. The exponent is only quadratic in $k$, and there are no large hidden constants. Note that we concentrate on the worst case in terms of the path-width. It is possible that the tree-width is significantly smaller than the path-width (by a factor of up to $\log n$). In this special case, the traditional approach can be faster.

\section{The Efficient Algorithm}
The crucial step of producing our faster path decomposition is to produce a small width path decomposition from one with a constant factor bigger width.

We are given a path decomposition $P$. Let
\[B^*_j = \bigcup_{i=1}^j B_i. \]
Let $G_j = G[B^*_j]$ be the subgraph of $G$ induced by $B^*_j$.

We want to construct a minimum width path decomposition $P'$ of $G$.
As for most efficient algorithms based on small \tw, our path decomposition algorithm uses a bottom-up dynamic programming approach. Any optimal path decomposition $P'$ of $G$ contains a path decomposition $P'_j$ of the subgraph $G_j$ of~$G$.

It is sufficient to try for all small $k$ whether there is a path decomposition of width $k$, and pick one in the affirmative case. For simplicity, we just describe the decision algorithm, because a solution can be found by standard back tracing of the dynamic programming solution.
The basic idea of the algorithm is to produce the skeleton $Q_j$ of an optimal $P'_j$ induced by $B_j$ for $j=1,\dots, 2n$. Then the path-width of $Q_{2n}$ is the path-width of $G$.

Naturally, the problem with this basic idea is that the optimal path decomposition $P'$ is unknown. A pessimistic approach would be to compute all skeletons obtained from all possible path decompositions of $G$. Fortunately, a good compromise is possible. Instead of computing all skeletons $Q_j$, we compute a set $\mathcal Q_j$ of skeletons, with the assurance that $\mathcal Q_j$ contains at least one skeleton $Q_j$ of an optimal path decomposition.

\begin{theorem} \label{thm:width-reduction}
Given a graph $G$, a number $k$, and a path decomposition $P$ of $G$ of width $\ell = 2^{O(k)}$, one can decide whether the \pw\ of $G$ is at most $k$
in time $2^{O(\ell k)} n$. A corresponding path decomposition can be computed with the same time bound.
\end{theorem}

\begin{proof}
 Assume $G$ and its path decomposition $P$ of width $\ell$ is given. Each bag $B_j$ is a set of vertices of $G$ with $|B_j| \leq \ell+1$.
We define $\mathcal A_j$ to be the set of all path decompositions of $G_j = G[B^*_j]$ of width at most $k$.

We will now define an algorithm that visits the nodes of $G$ in the order given by  $P$. For more details see Fig.\ \ref{alg:decrease-path-width}. The algorithm will have the following property.

\begin{claim}
During the visit of node $j$, the algorithm computes a set of skeletons $\mathcal Q_j$ of $\mathcal A_j$, induced by $B_j$. 
This set $\mathcal Q_j$ includes at least one of minimal width.
 \end{claim}
Each set of skeletons is computed from the previously computed set of skeletons of the predecessor $j-1$ in $P$. We will show that if $G_j$ has a path decomposition $P'_j$ of width at most $k' \leq k$, 
then $\mathcal Q_j$ contains at least one skeleton of $G_j$ (induced by $B_j$) of width at most $k'$.

\begin{figure}[thb] 
\begin{algorithm}{Decrease Path-Width}
\begin{tightlist}  
\item[Input:] 
A graph $G$, widths $\ell$ and $k$, and a path decomposition \\
 $P=((p_1,t_1,w_1),\dots,(p_{2n},t_{2n},w_{2n}))$ of $G$ of width $\ell$.
\item[Task:]
Decide whether $G$ has path-width at most $k$.
\item[Comment:]  Will be used with $\ell = 2k+1$.
\item[Comment:]  Compute a sequence of sets of skeletons ${\cal Q}_j$ of path decompositions of $G_j = G[B^*_j]$
 induced by the bag $B_j$ of node $j$, for $j = 1 ,\dots ,  2n$.
\item[Comment:]   For some optimal path decomposition $P'$ of $G$, each ${\cal Q}_j$ contains a skeleton of $P'[Q_j]$, the restriction of $P'$ to $B_j^*$.
\end{tightlist}
\(
\ind{0} \text{/\!/ Start Node} \\
\ind{0} w'_1 = 0; w'_2 = -1; {\cal Q}_1 = \{((p_1,+1,w'_1),(p_1,-1,w'_2))\}\\
\ind{0} \For\ j=2 \To\ 2n \Do \\
\ind{1} 	\text{/\!/ Introduce Node} \\
\ind{1}	\If\ t_j = +1 \Then\ \\
\ind{2}		\For\ \text{each skeleton $Q$ in ${\cal Q}_{j-1}$ create skeletons for ${\cal Q}_{j}$ by} \\
\ind{3}	\text{inserting $(p_j,+1,w'_{i}+1)$ in $Q$ after any position $i$,} \\
\ind{3}\text{inserting $(p_j,-1,w'_{i'}-1)$ in $Q$ after any position $i'$ with $i' > i$, and} \\
\ind{3}	\text{incrementing $w'_{i''}$ for all positions $i''$ in-between,} \\
\ind{3}	\text{as long as it does not increase the width above $k$}, \\
\ind{3}	\text{$i$ is before the forget node in $Q$ of any neighbor $u$ of $v$, and} \\
\ind{3}	\text{$i'$ is after the introduce node in $Q$ of any neighbor $u$ of $v$.} \\
\ind{3}	\text{(Positions include position 0 meaning insertion to the left of $Q$.)} \\
\ind{2}	\If\ {\cal Q}_{j} = \emptyset \Then\ \Return(\mbox{"tree-width $>k$"}). \\
\ind{1} 	\text{/\!/ Forget Node} \\
\ind{1}	\If\ t_j = -1 \Then\ \\
\ind{2}		\For\ \text{each skeleton $Q$ in ${\cal Q}_{j-1}$ create a skeleton for ${\cal Q}_{j}$ by} \\
\ind{3}	\text{replacing $p_j$ both times by 0, and simplifying as specified in Def.\ 				\ref{def:skeleton}.} \\
\ind{0} 	\text{/\!/ Success} \\
\ind{0}	\Return(\min\{\mbox{width of $Q$ : $Q \in {\cal Q}_{2n}$}\})
\)
\end{algorithm}
  \caption{The algorithm Decrease path-width}
  \label{alg:decrease-path-width}
\end{figure}

The algorithm Decrease Path-Width depends on the type of node $j$ in the path decomposition $P$. For every type of node, we now describe the action of the algorithm and prove the claim inductively.

\paragraph{Leaf node $1$:}
The bag of the leaf node $1$ contains just one vertex $v=p_1$, the only skeleton in $\mathcal Q_1$ is a two node skeleton $Q=(v,+1,0),(v,-1,-1)$ ($v$ is introduced and then forgotten). The leaf node 1 has bag $B_1=\{v\}$, the root node 2 has bag $B_2=\emptyset$.
The claim is satisfied, because there is only one path decomposition $P'$ of the one vertex graph $G_1$. Thus $Q$ is the skeleton of the minimal path decomposition $P'$.

\paragraph{Introduce node $j$:}
If $j$ is an introduce node, then the bag $B_j$ of node $j$ of $P$ contains a new vertex $v$ not in the bag $B_{j-1}$ of the predecessor node. More precisely, $B_j = B_{j-1} \cup \{v\}$ with $v \not\in B_{j-1}$.

The algorithm goes through all skeletons $Q$ of $\mathcal Q_{j-1}$. For each $Q$, it creates various skeletons of $\mathcal Q_j$ by inserting an introduce node and a forget node for the new vertex $v$. All suitable places are tried for the insertion of $v$. A place is suitable if it is before the forget nodes of all neighbors $u$ of $v$ in $G$.
Now the forget node of $v$ is inserted somewhere after the introduce node of $v$. This includes the option immediately after the introduce node of $v$ if that place is suitable. Again, all suitable places are tried. A place is suitable if it is after the introduce nodes of $v$ and all its neighbors $u$ in $G$.
A newly created skeleton is discarded rather than put into $\mathcal Q_j$, if its width is greater than $k$.

The algorithm would certainly be correct, if the insertions were tried in all positions of the full skeletons. This would be an extremely slow brute force algorithm. Thus it is crucial to argue that there is no benefit in trying those intervals $I$ of positions in the full skeletons that have been deleted in the proper skeletons. 
 Such an interval $I$ consists of the nodes $j''$ between two positions $j$ and $j'$ with 
$\min\{w'_{j},w'_{j'}\} \leq w'_{j''} \leq \max\{w'_{j},w'_{j'}\} $.
All the vertices introduced and deleted in these intervals are from $B_{i-1}^*$, while all vertices that still have to be included in the path decomposition are from bags after $i$. There are no edges between vertices of these intervals and later introduced vertices. Thus the only concern is the width caused by insertions between $j$ and $j'$.

The widths between positions $j$ and $j'$ can be viewed as a mountain range with the height at $j''$ being the current width $w'_{j''}$ of the bag $B'_{j''}$. If an insert or forget node of a later vertex is placed between positions $j$ and $j'$, it is not always an advantage to place these nodes in the deepest valley, because the width is also affected by mountain tops between the insertion and the deletion of a vertex.

Nevertheless, it is never an advantage to place insertions or deletions of later nodes along the slope of a mountain. All the later nodes can just as well, and often with an advantage, be placed at the bottom of a valley. 
From there, still the same mountain tops have to be crossed, but it can only be an advantage if the new width caused by later placed vertices is added to a smaller width (from a valley) of the earlier placed vertices. 
More precisely, if a highest and a lowest point of an interval $I$ are at its boundaries, and all the intermediate nodes insert and forget vertices of $B_{j-1}^*$, then there is never an advantage to insert or forget a new vertex at any other place on $I$ than at the lowest point. 

To prove the claim by induction, we assume that $\mathcal Q_{j-1}$ contains a skeleton $Q_{j-1}$ of width at most $k$ induced by $B_{j-1}$ corresponding to an optimal path decomposition $P^{(j-1)}$ of $G$. 
If $P^{j-1}$ has any vertices introduced or forgotten on the slopes of a skeleton $Q_{j-1}$, then $P^{(j-1)}$ is modified to a path decomposition $P^{(j)}$ by sliding all these vertices down the slope to the lowest valley of its interval. The width of $P^{(j)}$ is also optimal, because it is not more than the width of $P^{(j-1)}$. The skeleton $Q_j$ of this $P^{(j)}$ is in $\mathcal Q_j$ proving the induction step of the claim.

\paragraph{Forget node $j$:}
If $j$ is a forget node, then the bag $B_j$ of node $j$ of $P$ contains a new vertex $v$ not in the bag $B_{j-1}$ of the predecessor node. More precisely, $B_j = B_{j-1} \setminus \{v\}$ with $v \in B_{j-1}$. 
Now, $\mathcal Q_{j}$ is obtained from $\mathcal Q_{j-1}$ simply by restricting to the smaller set $B_j$, i.e., by replacing both occurrences of $v$ by 0.
This shows the correctness of the claim.

\paragraph{Running time:}
The running time of the algorithm is mainly determined by the number of  skeletons used. We have $O(n)$ nodes $j$ in the path decomposition $P$. For each node, we consider skeletons induced by the $\ell+1$ vertices from the bag $B_j$.
 By Proposition~\ref{prop:count}, we have $2^{O(\ell \log \ell)}$ path decompositions with $\ell+1$ vertices. 
 For each of these path decompositions, we have $2 \ell +1 $ intervals between the nodes where the vertices of the bag $B_j$ are inserted and deleted. These intervals are determined by their sequence of widths.
 
 The lengths of these intervals between two nodes involving vertices of $B_j$ in any skeleton are at most $2k+1$. This is so, because the worst width sequences in such a interval is $\dots ,k-2, 2, k-1, 1, k, 0, k, 1, k-1, 2, k-2, \dots$ and $\dots,k-2, 1, k-1, 0, k, 0, k-1, 1, k-2, \dots$. More importantly, there are only $2^{O(k)}$ such sequences \cite[Lemma 3.5]{BodlaenderK96}.
 
In summary, when handling bag $B_j$ for $j=1,\dots, 2n$, we have $2^{O(\ell \log \ell)}$ path decompositions of $B_j$. Each has $2 \ell +1 $ intervals with $2^{O(k)}$ possible sequences, resulting in $2^{O(\ell \log \ell)} 2^{O(\ell k)} = 2^{O(\ell k)}$ skeletons. 
 Thus for all $2n$ nodes of $P$ together, there are $2^{O(\ell k)}n$ possible skeletons. As the algorithm only makes polynomial time (in $k$) manipulations on each skeleton, the total running time is still $2^{O(\ell k)}n$.
 
 It is standard for dynamic programming algorithms to actually recover the structure (the path decomposition in our case) that has produced the minimum.
\qed \end{proof}

One should notice that there are no large hidden constants involved in the time analysis. 

\begin{corollary}
 The \pw\ of a graph can be computed in time $2^{O(k^2)}n$ for graphs of \pw\ $k$ if a path decomposition of width $O(k)$ is provided.
 \end{corollary}

\begin{proof}
 Theoretically, one can just try $k=1,2, \dots$ until one succeeds. Naturally, most of the work for $k-1$ could be used for $k$.
\qed \end{proof}

\section{Computing a Path-Width Approximation}

We use the well known result that in any tree decomposition of any graph containing the complete bipartite graph $K_{p,q}$, there is a bag containing either all the $p$ vertices of one side or all the $q$ vertices of the other side and an additional vertex. Thus if $p$ is greater than the tree-width $k$, then the addition of all edges between the $q$ vertices on the other side does not increase the tree-width or path-width. The graph obtained from $G$ by adding all such forced edges is called the augmented graph.

A vertex is simplicial, if its neighbors form a clique.

\begin{theorem}
 The path-width of a graph can be computed and a corresponding path decomposition can be found in time $2^{O(k^2)} n$ for graphs of path-width $k$.
 \end{theorem}

\begin{proof}
We use the results of \cite{BodlaenderK96,PerkovicR2000} that for any graph $G$ of tree-width $k$ one can quickly augment $G$ and find a linear size matching $M$ or a linear size subset $V'$ of simplicial vertices in the augmented graph. 

If $M$ is large, then one recursively computes an optimal path decomposition of width $k$ of the graph with the vertices of $M$ merged. It implies a path decomposition of the original graph $G$ of width at most $2k+1$. It can be improved to a path decomposition of width $k$ as seen in Theorem~{\ref{thm:width-reduction}}. 

If there is a large set $V'$ of simplicial vertices, then one recursively computes an optimal path decomposition of width $k$ of the graph with the vertices $V'$ removed.
If a tree decomposition of width $k$ is found for the graph with the simplicial vertices removed, then immediately such a decomposition can be obtained for the given graph. This does not work for path decompositions. One obtains a caterpillar graph instead.  
Then, we just change the caterpillar decomposition into a path decomposition of width 1 more. This width might not be optimal, but it is a very good approximation, that can be improved to an optimal path decomposition as before.
\qed \end{proof}

The reason why \cite{BodlaenderK96} works with tree decompositions, even when computing path decompositions, might be because the intermediate caterpillar decompositions prevent a direct production of path decompositions. Another reason is that for some graphs the tree-width is smaller than the path-width by a $\log n$ factor. In such situations, it could be better to work with tree-decompositions.

\begin{corollary}
  There is an algorithm computing the path-width, and outputting a corresponding path decomposition in time $2^{O(k^2)} n$, where $k$ is the path-width.
\end{corollary}

\begin{proof}
 As is usual for dynamic programming algorithms doing some minimization, whenever the algorithm makes a choice, computing the minimum width of two skeletons, it can record which option provided the minimum (an arbitrary choice is sufficient in case of a tie). Then it is easy to trace back to find an actual path decomposition once the global minimum width has been determined.
\qed \end{proof}

\section{Conclusion}
Path-width is an important width parameter. The worst case linear running time for its computation has not been improved over the last two decades. 
Our algorithm is significantly faster for path decomposition than the fastest linear time algorithms for tree decomposition. Furthermore, there are no large hidden constant factors in the expressions for the running time. We conjecture that this algorithm can be implemented to run satisfactory for small path-width. 

The main open problem in this area is to get an improvement of Bodlaender's algorithms of tree-width computation and the production of tree-decompositions of small width. In particular, one would like to know whether an $O(c^k n)$ time algorithm is possible for some constant $c$.
In a wider context the open question is whether similar results are possible for other important width parameters, in particular for clique-width.

\bibliographystyle{splncs03}

\end{document}